\documentclass[preprint,11pt]{elsarticle}

\usepackage{amssymb}
\newcommand{\ignore}[1]{}
\newcommand{\ignorecalc}[1]{}
\usepackage[english]{babel}
\usepackage{amsmath, amsthm, amssymb}
\usepackage{pslatex}
\usepackage{fancyhdr}
\usepackage{verbatim}
\usepackage{amssymb}
\usepackage{graphicx}
\usepackage{nomencl}
\usepackage{pstricks}
\usepackage{geometry}
\usepackage{url}

\usepackage[backref=page]{hyperref}
\hypersetup{
    unicode=false,          % non-Latin characters in Acrobat’s bookmarks
    pdftoolbar=true,        % show Acrobat’s toolbar?
    pdfmenubar=true,        % show Acrobat’s menu?
    pdffitwindow=true,      % page fit to window when opened
    pdftitle={Low Dimensional Euclidean Volume Preserving Embeddings},    % title
    pdfauthor={Anastasios Zouzias},     % author
    pdfsubject={Metric Embeddings},   % subject of the document
    pdfcreator={Anastasios Zouzias},   % creator of the document
    pdfproducer={Linux, Kile, Ubuntu, Macbook}, % producer of the document
    pdfkeywords={Metric Embedding, volume,dimensionality reduction, random, gaussian random matrix,distortion, determinants}, % list of keywords
    pdfnewwindow=true,      % links in new window
    colorlinks=false,       % false: boxed links; true: colored links
    linkcolor=red,          % color of internal links
    citecolor=green,        % color of links to bibliography
    filecolor=magenta,      % color of file links
    urlcolor=cyan          % color of external links
}

\renewcommand*{\backref}[1]{}
\renewcommand*{\backrefalt}[4]{%
	\ifcase #1 %
		(Not cited) %
	\or
		(Cited on page~#2)%
	\else 
		(Cited on pages~#2)%
	\fi
}

\usepackage{setspace}	% For double spacing.
\newcommand{\RR}{{\mathbb{R}}} % ETNA
\newcommand{\OO}{\mathcal{O}} % ETNA

\DeclareMathOperator{\vol}{Vol}

\renewcommand{\det}{\ensuremath{\text{det}}} % ETNA

\newcommand{\Prob}[1]{\ensuremath{\mathbb{P}\left(#1\right)}}

\newtheorem{theorem}{Theorem}

\newtheorem{lemma}{Lemma}

\newtheorem{remark}{Remark}

\newcommand{\Pset}{\mathcal{P}}

\journal{IPL}

\begin{document}

\begin{frontmatter}


\title{Low Dimensional Euclidean Volume Preserving Embeddings}
%\title{Volume Preserving Embeddings into Low Dimensional Euclidean Spaces}


\author{Anastasios Zouzias}
%\cortext[cor1]{Hello}
\ead{zouzias@cs.toronto.edu}
%\ead[url]{http://www.cs.toronto.edu/\textasciitilde zouzias/}
%%
\address{University of Toronto}
\begin{abstract}
Let $\Pset$ be an $n$-point subset of Euclidean space and $d\geq 3$ be an integer. In this paper we study the following question: What is the smallest (normalized) relative change of the volume of subsets of $\Pset$ when it is projected into $\RR^d$.
We prove that there exists a linear mapping $f:\Pset \mapsto \RR^d$ that relatively preserves the volume of all subsets of size up to $\lfloor d/2\rfloor$ within at most a factor of $\OO(n^{2/d}\sqrt{\log n \log\log n})$.
\end{abstract}
\end{frontmatter}
\section{Introduction}
A classical result of Johnson and Lindenstrauss~\cite{jl} states that any $n$-point subset of Euclidean space can be projected into $\OO(\log n)$ dimensions while preserving the metric structure of the set. A natural question to pose would be what is the smallest distortion of any $n$-point subset of Euclidean space when it is projected into (fixed) $d$ dimensions. This problem was first studied by Matou\v{s}ek~\cite{mat90}, who proved an $\OO(n^{2/d}\sqrt{\log n /d})$ upper bound on the distortion by projecting the points into $\RR^d$ using a random $d$-dimensional subspace. In Section~\ref{sec:distance} we re-prove Matou\v{s}ek's result using the simplified analysis of~\cite{jl-simple,im} adapted in this setting, i.e., bounding the distortion having fixed dimension instead of bounding the target dimension having fixed distortion. Although the simplified proof of the above result is well-known and well-understood, we hope that is not redundant and that it helps the reader to digest the following theorem
\begin{theorem}\label{thm:low-dim-vol}
Let $\Pset$ be a $n$-point subset of $\RR^N$ and let $3\leq d \leq c_3\log n$. Then there is a linear mapping $f:\Pset \mapsto \RR^d$ such that 
\begin{equation*}
	\forall S\subset \Pset, |S|\leq \lfloor d/2\rfloor \quad   1\leq \left(\frac{\vol(f(S))}{\vol (S)}\right)^{\frac1{|S| - 1}} \leq c_4 n^{2/d}\sqrt{\log n\log\log n},
\end{equation*}
where $c_3,c_4>0$ are absolute constants, and $\vol(S)$ is the $(|S|-1)$-dimensional volume of the convex hull of $S$.
\end{theorem}
\textbf{Remark:} The case where we fix the relative change of the volume of subsets to be arbitrary close to one, and ask what is the minimum dimension of such a mapping was studied in~\cite{MZ08}. 

Notice that if we only require to preserve pairwise distances the best upper bound is $\OO(n^{2/d}\sqrt{\log n/d})$, see Section~\ref{sec:distance}; therefore our result can be thought of as a generalization of the distance preserving embeddings since it also guarantees distance preservation. Moreover, there exists $n$-point subset of Euclidean space that any embedding onto $\RR^d$ has distortion $\Omega(n^{1/\lfloor (d+1)/2 \rfloor})$~\cite{mat90}, and thus the above worst-case upper bound cannot be much improved. 
%%%%%%%%%%%%%%%%%%%%%%%%%%%%%%%%%%%%%%%%%%%%%%%%%%%%%
%%%%%%%%%%%%%%%%%%%%%%%%%%%%%%%%%%%%%%%%%%%%%%%%%%%%%
%%%%%%%%%%%%%%%%%%%%%%%%%%%%%%%%%%%%%%%%%%%%%%%%%%%%%
%%%	Preliminaries and Technical Lemmas
%%%%%%%%%%%%%%%%%%%%%%%%%%%%%%%%%%%%%%%%%%%%%%%%%%%%%
%%%%%%%%%%%%%%%%%%%%%%%%%%%%%%%%%%%%%%%%%%%%%%%%%%%%%
%%%%%%%%%%%%%%%%%%%%%%%%%%%%%%%%%%%%%%%%%%%%%%%%%%%%%
\section{Preliminaries and Technical Lemmas}\label{sec:defn}
We start by defining an (stochastic) ordering between two random variables $X$ and $Y$, but first let's motivate this definition. Assume that we have upper and lower bounds on the distribution function of $Y$, and also assume that it's hard to give precise bounds on the distribution function of $X$. Using this notion of ordering, if $X$ ``smaller than'' $Y$, then we can bound the ``complicated'' variable $X$ through bounding the ``easy'' variable $Y$. We use this notion extensively in this paper.

More formally, let $X$ and $Y$ be two random variables, not necessarily on the same probability space. The random variable $X$ is stochastically smaller than the random variable $Y$ when, for every $x\in{\RR}$, the inequality 
\begin{equation}\label{eq:rv:ordering}
\Prob{X\leq x} \geq \Prob{Y\leq x}
\end{equation}
holds. We denote this by $X\preceq Y$.

Next we recall known results about the Chi-square distribution and also give bounds on its' cumulative distribution function. If $X_i,\ i=1,\dots ,d$ be independent, identically distributed normal random variables, then the random variable $\chi_d^2=\sum_{i=1}^{d}{X_i^2}$ is a Chi-square random variable with $d$ degrees of freedom. Notice that the expected value of $\chi_d^2$ is $d$. It is well known~\cite[Chapter~II, p.~$47$]{book:feller:vol2} that the Chi-square distribution is a special case of the Gamma distribution and its cumulative distribution function is given by
\begin{equation}
\Prob{\chi_d^2 \leq t}= \frac{\gamma(d/2,t/2)}{\Gamma(d/2)},
\end{equation}
where $\Gamma(x)$ is the Gamma function, $\gamma(a,x)=\int_{0}^{x} {t^{a-1} e^{-t}\, dt}$ and $\Gamma(a,x)=\int_{x}^{\infty} {t^{a-1} e^{-t}\, dt}$  is the lower and upper \emph{incomplete Gamma function}, respectively. Next we present some bounds on the Gamma and incomplete Gamma functions that we use in Sections~\ref{sec:distance}, \ref{sec:volume}. We start by presenting the following bound on the Gamma function, see for instance~\cite[Lemmas~$2.5,2.6,2.7$]{condRandomMatrix} and~\cite[p.$253$]{book:WW}.
\begin{lemma}[Stirling Bound on Gamma Function]\label{lem:Gamma}
If  $\Gamma(a)=\int_{0}^{\infty}{e^{-t}t^{a-1}\,dt}$, where $a>0$, then 
\begin{equation}
 \sqrt{2\pi} a^{a+1/2} e^{-a} < \Gamma (a+1) < \sqrt{2\pi} a^{a+1/2}e^{-a+\frac1{12a}},
\end{equation}
\end{lemma}
Next we upper bound $\gamma(a,x)$. Note that $\gamma(a,x)=\int_{0}^{x} {t^{a-1} e^{-t}\, dt}\leq \int_{0}^{x} {t^{a-1}\, dt}$, hence
\begin{equation}\label{eq:incgamma}
 \gamma(a,x) \leq x^a/a.
\end{equation}
Now for the upper incomplete gamma, we have the following bound.
\begin{lemma}\label{lem:incGamma}
 If $\Gamma(a,x)= \int_{x}^{\infty}{e^{-t} t^{a-1}\, dt}$ where $x>2(a+1)$, then 
\begin{equation}
\Gamma(a, x)  < 2\exp(-x)x^{a+1}.
\end{equation}
\end{lemma}
\begin{proof}
In \cite[Lemma~$2.6$]{condRandomMatrix} set $\alpha =1$ and $d=2$. 
\end{proof}
It is well-known~\cite[pp.~$220-235$]{book:FB_90} that the volume that is spanned by the convex hull of a $k$-point subset of $\RR^N$ along with the origin is equal to $\sqrt{\det( P^\top P)} / k!$, where $P$ is the $k\times N$ matrix that contains the points as columns. The following lemma gives a connection between the volume of the convex hull of $k$ points and the determinant of a specific matrix that is constructed using these points.
\begin{lemma}
Let $\mathcal{P}=\{p_1,p_2,\dots, p_k\}$ be an $k$-point subset of $\RR^N$  in general position and let $f:\RR^N \mapsto \RR^d$ be a linear mapping. Let $P:= [p_2 - p_1, p_3 -p_1 ,\dots ,p_k - p_1]$ be an $N\times (k-1)$ matrix. Then
\begin{equation}
 \frac{\vol (f(\mathcal{P}) )}{\vol (\mathcal{P})}= \left(\frac{\det \left( (FP)^\top FP\right)}{\det (P^\top P) } \right)^{1/2},
\end{equation}
where $F$ is the $d\times N$ matrix that corresponds to $f$.
\end{lemma}
\begin{proof}
By a translation of the point-set $\Pset$, i.e., identifying $p_1$ with the origin, it follows that $\vol (\mathcal{P})= \sqrt{\det( P^\top P)}/k!$, since the volume is translation invariant, and similarly $\vol (f(\mathcal{P}))= \sqrt{\det( (FP)^\top FP)}/k!$. Since $\mathcal{P}$ is in general position, it follows that
\[\frac{\vol (f(\mathcal{P}) )}{\vol (\mathcal{P})}= \left(\frac{\det \left( (FP)^\top FP\right)}{\det (P^\top P) } \right)^{1/2}.\]
\end{proof}
Now, let's consider the above lemma in the setting where $f$ is a random linear mapping. More specifically, let $F$ be a Gaussian matrix, i.e., a matrix whose entries are i.i.d. Gaussian $\mathcal{N}(0,1)$. First observe that the fraction of the volumes is a random variable. Surprisingly enough, as the following lemma states, the fraction of the volumes in this setting is \emph{independent} of $\Pset$. This can be thought of as a generalization of the $2$-stability property of inner products with Gaussian random vectors to matrix multiplication with Gaussian matrices.
\begin{lemma}\label{lem:vol_stability}
Let $\mathcal{P}=\{p_1,p_2,\dots, p_k\}$ be an $k$-point subset of $\RR^N$ in general position. And let $f:\RR^N \mapsto \RR^d$ be a random Gaussian linear mapping. Then
\begin{equation}\label{eq:frac_volumes}
 \left(\frac{\vol (f(\mathcal{P}) )}{\vol (\mathcal{P})} \right)^2 \sim \prod_{i=1}^{k - 1}{\chi^2_{d - i + 1}}.
\end{equation}
\end{lemma}
\begin{proof}
It is a simple consequence of \cite[Lemma~$3$]{MZ08} and the above lemma.
\end{proof}
\begin{remark}\label{rem:distance}
For $k=2$ in Lemma~\ref{lem:vol_stability}, we get $\|f(p_1)-f(p_2)\|^2/\|p_1-p_2\|^2 \sim \chi_d^2$. 
\end{remark}

Equation~\ref{eq:frac_volumes} gives the distribution of the fraction of the volume as a product of independent random variables. However, in general it's difficult to deal with such a product, and so we employ the following theorem that sandwiches this product with a single Chi-square distributions.
\begin{theorem}[Theorem~4,~\cite{gordon}]\label{thm:gordon}
Let $u_i:=\chi^2_{d-i+1}$ be independent Chi-square random variables for $i=1,2, \dots , s$. Then the following holds for every $s\geq 1$, 
\begin{equation}
\chi^2_{s(d-s+1) +\frac{(s-1)(s-2)}{2}} \succeq s \left(\prod_{i=1}^{s}{u_i}\right)^{1/s} \succeq \chi^2_{s(d-s+1)} .
\end{equation}
\end{theorem}
We now have enough tools at our disposal to prove Theorem~\ref{thm:low-dim-vol}.
%%%%%%%%%%%%%%%%%%%%%%%%%%%%%%%%%%%%%%%%%%%%%%%%%%%%%
%%%%%%%%%%%%%%%%%%%%%%%%%%%%%%%%%%%%%%%%%%%%%%%%%%%%%
%%%%%%%%%%%%%%%%%%%%%%%%%%%%%%%%%%%%%%%%%%%%%%%%%%%%%
%%%			(Distance) Distortion
%%%%%%%%%%%%%%%%%%%%%%%%%%%%%%%%%%%%%%%%%%%%%%%%%%%%%
%%%%%%%%%%%%%%%%%%%%%%%%%%%%%%%%%%%%%%%%%%%%%%%%%%%%%
%%%%%%%%%%%%%%%%%%%%%%%%%%%%%%%%%%%%%%%%%%%%%%%%%%%%%
\section{Distance Distortion}\label{sec:distance}
In this section we prove the following
\begin{theorem}\label{thm:mat90:distances}
Let $\Pset$ be a $n$-point subset of $\RR^N$ and let $3\leq d \leq c_1 \log n$, where $c_1$ is a positive constant. Then there exists a linear mapping $f:\Pset \mapsto \RR^d$ with (distance) distortion $\text{dist}(f)=\OO(n^{2/d}\sqrt{\log n/d})$, i.e., there exists an absolute constant $c>0$ such that 
\[\forall x,y \in \Pset, \qquad \|x-y\| \leq \|f(x)-f(y)\| \leq cn^{2/d}\sqrt{\log n/d}\|x-y\|. \]
\end{theorem}
\begin{proof}
Similarly as in~\cite{mat90}. Consider the random linear map $f:\RR^N \to \RR^d$, $f(x):=R\cdot x$ where $R$ is an $d\times N$ random Gaussian matrix. Using linearity of $f$ and Remark~\ref{rem:distance} it follows that $\|f(x)-f(y)\|^2/\|x-y\|^2\sim \chi_d^2$ for any $x,y\in{\Pset}$. Our goal is to show that $\chi_d^2$ is sufficiently concentrated. More specifically, it suffices to show that $\chi_d^2$ doesn't fall outside an interval $[a,b]$ for some $a,b\in\RR$ with constant probability. This aims to upper bound the probabilities $Pr[\chi_d^2 \leq a^2]$ and  $Pr[\chi_d^2 \geq b^2]$.
\par The elements of $\Pset$ determine at most $\binom{n}{2}$ distinct direction vectors. Applying union bound over all pairs of $\Pset$ gives that if 
\begin{equation}\label{ineq:distance:unionbound}
 \binom{n}{2} \left( \Prob{\chi_d^2 \leq a^2} + \Prob{\chi_d^2 \geq b^2} \right) < 1,
\end{equation}
then there exists $f$ that expands every distance in $\Pset$ by at most $b$ times and contracts at least $a$ times, so $\rm{dist}(f)\leq b/a$. Our goal therefore is to specify $a,b$ in terms of $d$ and $n$ such that Inequality~\ref{ineq:distance:unionbound} holds.
To do so, we first bound $\Gamma(d/2)$ from below, which will be used later. By Lemma~\ref{lem:Gamma}, we have that $\Gamma (d/2) \geq  e^{- d/2}(d-2)^{(d-1)/2} /2^{d/2}.$
\ignorecalc{\begin{eqnarray*}
\Gamma (d/2) &>& \sqrt{2\pi} (d/2-1)^{d/2 -1+1/2} e^{-(d/2-1)} \\
\ignorecalc{& >& \sqrt{2\pi} (\frac{d-2}{2})^{d/2 -1/2} e^{- d/2} \\
& > & \sqrt{2\pi} \frac{e^{- d/2}(d-2)^{(d-1)/2}}{2^{d/2-1/2}} \\}
& > & \frac{e^{- d/2}(d-2)^{(d-1)/2} }{2^{d/2}} .
\end{eqnarray*}
}
Now, we will bound $a,b$ separately. We find $a$ such that $\binom{n}{2} \Prob{\chi_d^2 \leq a^2} < 1/2$. Using Equation~\ref{eq:incgamma} and the previous analysis we require that $\frac{n^2}{2} \frac{a^d}{e^{-d/2} (d-2)^{(d-1)/2}} < 1/2$, which holds for all $d\geq 3$ if we set $a=c_2\sqrt{d}/n^{2/d}$, where $c_2>0$ is an absolute constant.
\ignore{
Which holds if 
\begin{eqnarray*}
a &<& \frac{e^{-1/2} (d-2)^{\frac{d-1}{2d}}}{n^{2/d}} \\
&<&  \frac{(d-2)^{1/2}}{n^{2/d}} \\
&=& \OO(\frac{\sqrt{d}}{n^{2/d}}).
\end{eqnarray*}
}
Similarly, we will find $b$ such that $\binom{n}{2} \Prob{\chi_d^2 \geq b^2} < 1/2$. Using Lemma~\ref{lem:incGamma}, and assume for the moment that $b^2 > 2d-2 $, we have that
\ignorecalc{\begin{eqnarray*}
\Prob{\chi_d^2 \geq b^2} &\leq &  \frac{e^{-b^2/2} (b^2/2)^{d/2-1}}{\Gamma(d/2)} \\
\ignorecalc{&\leq & \frac{2e^{-b^2/2} b^{d-2} }  {2^{d/2} e^{d/2} (d-2)^{(d-1)/2}} \\}
&\leq & \frac{b^{d-2} e^{-b^2/2 - d/2}}{(d-2)^{(d-1)/2}}. 
\end{eqnarray*}
}
\begin{equation*}
\Prob{\chi_d^2 \geq b^2}\ \leq\  \frac{e^{-b^2/2} (b^2/2)^{d/2-1}}{\Gamma(d/2)}\ \leq\ \frac{b^{d-2} e^{-b^2/2 - d/2}}{(d-2)^{(d-1)/2}}.
\end{equation*}
It suffices to show that $\ln \left(n^2 \frac{b^{d-2} e^{-b^2/2 - d/2}}{(d-2)^{(d-1)/2}} \right)$ is negative for large enough $n$. Indeed,
\begin{eqnarray*}
\ln \left(n^2 \frac{b^{d-2} e^{-b^2/2 - d/2}}{(d-2)^{(d-1)/2}} \right) &< & 2\ln n + (d-2)\ln b  - b^2/2 -d/2 - \frac{d-1}{2} \ln (d-2).
\end{eqnarray*}
Note that if $d>d'$ then $\Prob{\chi_{d'}^2 \geq b^2} \leq \Prob{\chi_d^2 \geq b^2} $. Thus we can assume that $d=c_1 \log n$, since if we can bound it, then we can bound it for all fixed $d< c_1 \log n$. Define $g(b,n)=2\ln n + (d-2)\ln b  - b^2/2 -d/2 - \frac{d-1}{2} \ln (d-2)$. We want to show that $g(b,\ n)<0$ for large enough $n$. By choosing $b=5c_1 \sqrt{\log n}$, and recall that $d=c_1 \log n$ hence $b^2 >2d-2$, we conclude that $\lim_{n\to \infty}g(5\sqrt{\ln n},\ n)= -\infty$ as desired.
Hence, we can choose $a,b$ functions of $n$ such that $b/a = \frac{5c_1\sqrt{\log n} }{c_2\sqrt{d}/n^{2/d}} =cn^{2/d} \sqrt{\log n /d}$.
\end{proof}
%%%%%%%%%%%%%%%%%%%%%%%%%%%%%%%%%%%%%%%%%%%%%%%%%%%%%
%%%%%%%%%%%%%%%%%%%%%%%%%%%%%%%%%%%%%%%%%%%%%%%%%%%%%
%%%%%%%%%%%%%%%%%%%%%%%%%%%%%%%%%%%%%%%%%%%%%%%%%%%%%
%%%			Volume Distortion
%%%%%%%%%%%%%%%%%%%%%%%%%%%%%%%%%%%%%%%%%%%%%%%%%%%%%
%%%%%%%%%%%%%%%%%%%%%%%%%%%%%%%%%%%%%%%%%%%%%%%%%%%%%
%%%%%%%%%%%%%%%%%%%%%%%%%%%%%%%%%%%%%%%%%%%%%%%%%%%%%
\section{Proof of Main Theorem}\label{sec:volume}
Our goal is to find a mapping $f:\Pset \to \RR^d$ such that
\begin{equation}\label{eq:low-volume-constraints}
\forall S\subset \Pset, |S|\leq k \quad   1 \leq \left(\frac{\vol(f(S))}{\vol (S)} \right)^{\frac{1}{|S| - 1}} \leq D,
\end{equation}
where $D$ is the volume distortion of the mapping. We will see in the analysis below that we can set $k=\lfloor d/2 \rfloor$ and $D=\OO(n^{2/d}\sqrt{\log n\log\log n})$. We can assume w.l.o.g. that the input points are in general position, i.e., every subset of size up to $k$ is affinely independent. If not, both the original points and projected points will span zero volume.

Similarly with Section~\ref{sec:distance}, we take a random $f$ using a Gaussian random matrix and show that it satisfies~(\ref{eq:low-volume-constraints}) with constant probability. To do so, we first bound the probability that a fixed subset ``contracts'' its' volume by more than a factor $a$.
\begin{lemma}
Fix any subset $S\subset \Pset$ of size $|S|=s+1$ with $1\leq s < k$. Then
\[\Prob{ \left(\frac{\vol(f(S))}{\vol(S)}\right)^{\frac{1}{|S|-1}} \leq a}\ \leq\ \frac{(esa^2)^{t/2}}{t(t-2)^{(t-1)/2}}, \]
where $t=s(d-s+1)$.
\end{lemma}
\begin{proof}
Using Lemma~\ref{lem:vol_stability} we know that the above probability is equal to $\Prob{ \left(\prod_{i=1}^{s} \chi_{d-i+1}^2\right)^{1/s} \leq  a^2}$. Using Theorem~\ref{thm:gordon}, we can bound the above probability of product of Chi-square random variables with a single Chi-square. More specifically, using the stochastic ordering we have the following inequality
\[\Prob{ \left(\prod_{i=1}^{s} \chi_{d-i+1}^2\right)^{1/s} \leq  a^2} \leq \Prob{\chi_{s(d-s+1)}^2 \leq s\cdot a^2}\]
for every $1\leq s < k$.
Now, we have a single Chi-square random variable and thus we can bound it from above, the same way as we did in Section~\ref{sec:distance}, using Lemma~(\ref{lem:Gamma}) and Equation (\ref{eq:incgamma}). It follows that $\Prob{\chi_{t}^2 \leq s\cdot a^2}\ =\ \frac{\gamma (t/2,sa^2/2)}{\Gamma (t/2)}\ \leq\ \frac{(esa^2)^{t/2}}{t(t-2)^{(t-1)/2}}$.
\end{proof}
Similarly, we bound the probability that a fixed subset ``expands'' it's volume by more than a factor $b$.
\begin{lemma}\label{lem:volume_expansion}
Fix any subset $S\subset \Pset$ of size $|S|=s+1$ with $1\leq s < k$. If $sb^2 > 2l+4$, then
\begin{equation*}
\Prob{\left(\frac{\vol(f(S))}{\vol(S)}\right)^{\frac{1}{|S|-1}} \geq b}\ \leq\ \frac{e^{-\frac{sb^2-l}{2}}(sb^2)^{l/2+1} }{ (l-2)^{(l-1)/2}},
\end{equation*}
where $l=s(d-s+1)+\frac{(s-1)(s-2)}{2}$.
\end{lemma}
\begin{proof}
As in the previous lemma the above probability is equal to $\Prob{ \left(\prod_{i=1}^{s} \chi_{d-i+1}^2\right)^{1/s}\geq b^2}$, and again using Theorem~\ref{thm:gordon} it follows that
\begin{equation*}
	\Prob{ \left(\prod_{i=1}^{s} \chi_{d-i+1}^2\right)^{1/s} \geq b^2} \leq \Prob{\chi_{s(d-s+1)+ \frac{(s-1)(s-2)}{2}}^2 \geq s\cdot b^2}:=E_{d,s}.
\end{equation*}
Using Lemmas~\ref{lem:Gamma},~\ref{lem:incGamma} it follows that $\Prob{\chi_{l}^2 \geq s\cdot b^2}\ =\ \frac{\Gamma(l/2,sb^2/2)}{\Gamma(l/2)}\ \leq\ \frac{e^{-\frac{sb^2 - l}{2}}(s b^2)^{l / 2 + 1 } }{(l - 2)^{(l - 1) / 2}}$. 
\end{proof}
Notice that if $d'>d$, then $E_{d,s}\leq E_{d',s}$ from the stochastic ordering of the Chi-square distribution. Now we are ready to apply union bound. Our goal is to find $a$ such that with probability at least $1/2$, our embedding does not contract volumes of subsets of size up to $k$ by a factor $a$.

By union bounding over all sets of fixed size $i$, $1\leq i \leq k$, we want to find $a$ such that 
\begin{eqnarray*}
\binom{n}{i+1} \frac{(eia^2)^{t_i/2}}{t_i(t_i-2)^{(t_i-1)/2}} < \frac1{2k},
\end{eqnarray*}
where $t_i=i(d-i+1)$. Note that if we sum over all different size of subsets ($i=1,\dots ,k$) we get that the failure probability is at most $1/2$. It suffices to show that $\ln \left(2k \binom{n}{i+1}  \frac{(eia^2)^{t_i/2}}{t_i(t_i-2)^{(t_i-1)/2}}\right)$ is negative for large enough $n$ and for every $1\leq i\leq k$ and $d\geq 3$, or equivalently the following is negative
\begin{eqnarray*}
\ln 2 +\ln k +(i+1)\ln n +t_i\ln a &+& (t_i/2-i)\ln i + (t_i/2 + i)-\ln t_i - (\frac{t_i-1}{2})\ln (t_i-2).
\end{eqnarray*}
Let's group the terms of the right hand size and bound them individually. It is not hard to see that $(t_i/2-i)\ln i - (\frac{t_i-1}{2})\ln (t_i-2) < 0$ and $\ln k - \ln t_i \leq  0$ since $k \leq d \leq t_i$ and $t_i =i(d-i+1)$, when $i=1,\dots , k$ and for $d\geq 3$. Hence, it suffices to show that
\begin{eqnarray*}
\ln 2 +(i+1)\ln n +t_i\ln a  + (t_i/2 + i) &<& 0.
\end{eqnarray*}
Set $a=c_e n^{-\gamma}$, for some positive $\gamma$ that will be specified shortly and $c_e$ a sufficient small positive constant. Recall that we want the above inequality to hold for every $1\leq i \leq k$. We can choose $c_e$ smaller than $e^{-1}$ and take care of the $t_i/2+i+\ln 2$ term. Lets now focus on the dominate term $(i+1)\ln n$. It follows that the above quantity is negative if $\gamma\ \geq\ \frac{i+1}{i(d-i+1)},\ \text{for all }i=1,\dots, k$.
Let's study closer the function $h_d(x)=\frac{x+1}{x(d-x+1)}$. We will show that $h_d(x)$ is convex on the domain $[1,d/2]$ and also increasing in the domain $[d/4,d]$ for any fixed $d\geq 3$.
\ignore{The first and second derivatives of $h_d(x)$ are :
\begin{eqnarray*}
 h_d'(x) &=& \frac{x^2+2x-d-1}{x^2(d-x+1)^2}\\
 h_d''(x) &=& \frac{2(x^3+3x^2-3dx -3x +d^2 +2d +1)}{x^3(d-x+1)}.
\end{eqnarray*}
}
A simple calculation shows that $h_d''(x) > 0$ for $x\in{[1,d]}$ and $h_d'(x)>0$ for $x\in{[\frac{d}{4},d]}$ (details omitted). Also note that $h_d(1)=h_d(d/2)=2/d$. By convexity in $[1,d/2]$, we get that $h_d(x)\leq 2/d$ for all $x\in{[1,d/2]}$.
\par The above analysis gives a bound on the parameter $k$, i.e., the maximum size of subsets that we can consider. Thus, we get that $k$ should be less than or equal to $\lfloor d/2\rfloor$.
\par To sum up, we have proved that if $a=c_en^{-2/d}$ then with probability at least $1/2$ our embedding doesn't contract the normalized volumes of subsets of size at most $\lfloor d/2\rfloor$ by more than a multiplicative factor of $a$.
\ignore{, i.e.,
\[\Prob{\forall S\subset P, |S|\leq \lfloor d/2\rfloor,\left(\frac{\vol(f(S))}{\vol(S)}\right)^{\frac{1}{|S|-1}} \geq a }> \frac1{2}. \]
}

Next our goal is to find $b$ such that with probability at least $1/2$, $f$ does not expand volumes by more than a factor of $b$. Let $l_i=i(d - i + 1) + \frac{(i - 1)(i - 2)}{2}$. We apply union bound over all sets of fixed size $i$, $1\leq i < k$ together with Lemma~\ref{lem:volume_expansion} assuming for the moment that $ib^2 > 4l_i+8$. We want to find $b$ such that 
\begin{eqnarray*}
\binom{n}{i+1}  \frac{e^{-\frac{ib^2 - l_i }{2}}(i b^2)^{l_i / 2 + 1} }{(l_i - 2)^{( l_i - 1) / 2}}< \frac1{2k}.
\end{eqnarray*}
Summing over all different size of subsets we get the desired property with probability at least $1/2$.
\par It suffices to show that $\ln \left(2k \binom{n}{i+1} \frac{e^{-\frac{i b^2 - l_i}{2}}(i b^2)^{l_i / 2 + 1} }{ (l_i - 2)^{( l_i - 1) / 2}}\right)$ is negative for every $1\leq i < k$ and $d\in{[3,\log n]}$. Similarly with Section~\ref{sec:distance} we can assume without loss of generality that $d=c_3\log n$, using the fact that if $d'\leq d$ then $E_{d',s}\leq E_{d,s}$.

Now, since there are at most $\binom{n}{i+1}\leq \left(\frac{ne}{i+1}\right)^{i+1}$ subsets of size $i+1$, it suffices to show that the following quantity is negative,
\begin{eqnarray*}
\ln \left( \frac{k n^{i+1}e^{-\frac{i b^2 - l_i - 2 i}{2}}(i b^2)^{l_i / 2 + 1} (i + 1)^{-(i + 1)} }{ (l_i - 2)^{(l_i - 1) / 2}}\right)  \leq  \ln \left( \frac{k n^{i + 1} e^{-\frac{i b^2 - l_i - 2 i}{2}}i^{ l_i / 2 - i}b^{l_i + 2} }{(l_i - 2)^{(l_i - 1) / 2}}\right) <\\
\ignorecalc{\ln k + (i + 1) \ln n - \frac{i b^2 - l_i - 2 i}{2} + (l_i/2 - i)\ln i + (l_i - 2) \ln b - ( \frac{l_i - 1}{2})\ln (l_i - 2) <\\}
(l_i/2 + 1)\ln i + (i + 1) \ln n  + l_i  \ln b + l_i/2 + 2i +  \ln k -\left(\frac{i b^2}{2} + \frac{l_i - 1}{2}\ln l_i  \right).
\end{eqnarray*}
Note that in the last quantity the positive terms are of order $\OO(i d\ln i + i\ln n)$. The negative terms are of order $\OO(ib^2)$.  Recall that $i<d=c_3 \log n $. It is not hard to see that by choosing $b=c_2 \sqrt{\log n\log\log n}$, where $c_2>0$ a sufficient large constant, then $ib^2 > 4l_i +8$ and the above quantity goes to $-\infty$ as $n$ grows for every $1 \leq i < k$.
\par To sum up, we proved that with probability at least $1/2$, $f$ doesn't expand normalized volumes of subsets of size at most $\lfloor d/2 \rfloor$ by more than a multiplicative factor of $b$.
\ignore{, i.e.,
\[\Prob{\forall S\subset P, |S|\leq k, \left(\frac{\vol(f(S))}{\vol(S)}\right)^{\frac{1}{|S|-1}} \leq b} > \frac1{2}. \]}

Rescaling $f$ by $a$, we conclude that there exists $a,b$ with $a<b$ such that
\begin{eqnarray*}
\Prob{\forall S\subset P, |S|\leq \lfloor d/2 \rfloor, 1\leq \left(\frac{\vol(f(S))}{\vol(S)}\right)^{\frac{1}{|S|-1}} \leq \frac{b}{a}} > 0.
\end{eqnarray*}
This concludes the proof of Theorem~\ref{thm:low-dim-vol}.
%\bibliographystyle{plain}
{\small
\bibliographystyle{alpha}
\bibliography{refs.bib} 
}
\end{document}